%% file: main.tex
\documentclass{llncs}

\usepackage[dvips]{graphicx}
\usepackage{graphicx}
\usepackage{amssymb}
\usepackage{amsmath}
\usepackage{enumerate}

\input{newcom.tex}

\title{The List Coloring Reconfiguration Problem\\ for Bounded Pathwidth Graphs}

\author{
	Tatsuhiko Hatanaka
\and
	Takehiro Ito
\and
	Xiao Zhou
}

\institute{
	Graduate School of Information Sciences, 
	Tohoku University, \\
	Aoba-yama 6-6-05, Sendai 980-8579, Japan.\\
	\email{\{hatanaka, takehiro, zhou\}@ecei.tohoku.ac.jp}
}

\begin{document}
\maketitle

\begin{abstract}
	We study the problem of transforming one list (vertex) coloring of a graph into another list coloring by changing only one vertex color assignment at a time, while at all times maintaining a list coloring, 	given a list of allowed colors for each vertex.
	This problem is known to be PSPACE-complete for bipartite planar graphs.
	In this paper, we first show that the problem remains PSPACE-complete even for bipartite series-parallel graphs, which form a proper subclass of bipartite planar graphs. 
	We note that our reduction indeed shows the PSPACE-completeness for graphs with pathwidth two, and it can be extended for threshold  graphs.
	In contrast, we give a polynomial-time algorithm to solve the problem for graphs with pathwidth one. 
	Thus, this paper gives precise analyses of the problem with respect to pathwidth. 
%
\end{abstract}

\section{Introduction}

	Graph coloring is one of the most fundamental research topics in the field of theoretical computer science. 
	Let $\colorset = \{ \cind{1}, \cind{2}, \ldots, \cind{\numk} \}$ be the set of $\numk$ colors.
	A (proper) \textit{$\numk$-coloring} of a graph $G = (V,E)$ is a mapping $\mapf : V \to \colorset$ such that $\mapf(v) \neq \mapf(w)$ for every edge $vw \in E$.
	In {\em list coloring}, each vertex $v \in V$ has a set $\verlist(v) \subseteq \colorset$ of colors, called the \textit{list} of $v$.
	Then, a $\numk$-coloring $\mapf$ of $G$ is called a \textit{$\numk$-list coloring} of $G$ if $\mapf(v) \in \verlist(v)$ holds for every vertex $v \in V$.
	Figure~\ref{fig:example}(b) illustrates four $\numk$-list colorings of the same graph $G$ with the same list $L$ depicted in \figurename~\ref{fig:example}(a); the color assigned to each vertex is attached to the vertex.
	Clearly, a $\numk$-coloring of $G$ is a $\numk$-list coloring of $G$ for which $\verlist(v) = \colorset$ holds for every vertex $v$ of $G$, and hence $\numk$-list coloring is a generalization of $k$-coloring. 

	Graph coloring has several practical applications, such as in scheduling, frequency assignments. 
	For example, in the frequency assignment problem, each vertex corresponds to a base station and each edge represents the physical proximity and hence the two corresponding base stations have the high potential of interference.  
	Each color represents a channel of a particular frequency, and we wish to find an assignment of channels to the base stations without any interference. 
	Furthermore, in list coloring, each base station can have a list of channels that can be assigned to it.

	\subsection{Our problem}

	However, a practical issue in channel assignments requires that the formulation should be considered in more dynamic situations.
	One can imagine a variety of practical scenarios where a $\numk$-list coloring (e.g., representing a feasible channel assignment) needs to be transformed (to use a newly found better solution or to satisfy new side constraints) by individual color changes (keeping the network functionality and preventing the need for any coordination) while maintaining feasibility (so that the users receive service even during the reassignment).

	In this paper, we thus study the following problem: 
	Suppose that we are given two $\numk$-list colorings of a graph $G$ (e.g., the leftmost and rightmost ones in \figurename~\ref{fig:example}(b)), and we are asked whether we can transform one into the other via $\numk$-list colorings of $G$ such that each differs from the previous one in only one vertex color assignment.
	We call this decision problem the {\sc $\numk$-list coloring reconfiguration} problem. 
	For the particular instance of \figurename~\ref{fig:example}(b), the answer is ``yes,'' as illustrated in \figurename~\ref{fig:example}(b), where the vertex whose color assignment was changed from the previous one is depicted by a black circle. 

\begin{figure}[t]
\begin{center}
\includegraphics[width=\textwidth]{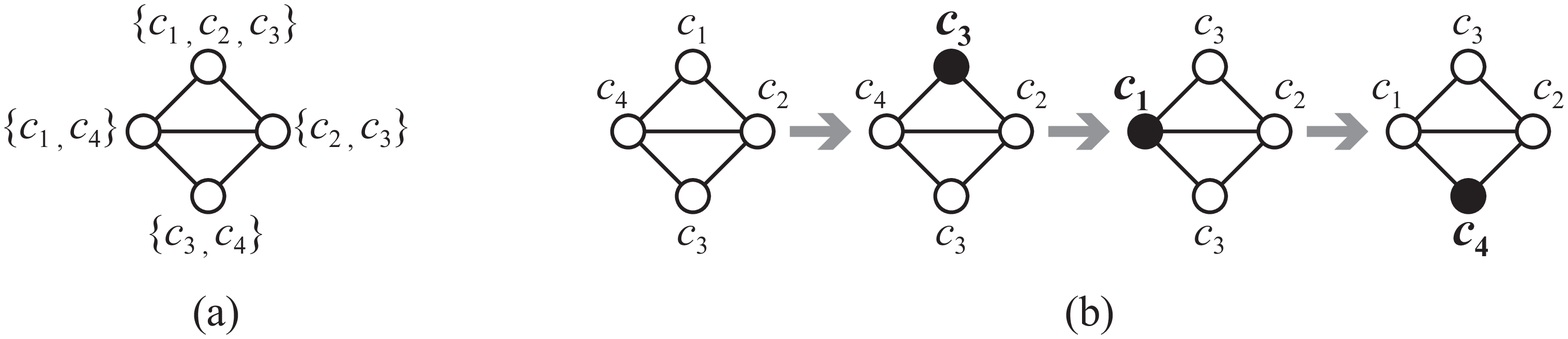}
\end{center}
\vspace{-2em}
\caption{(a) Graph $G$ and its list $L$, and (b) a sequence of $\numk$-list colorings of $G$.}
\label{fig:example}
\vspace{-1em}
\end{figure}

\subsection{Known and related results}
	
		Recently, similar settings of problems have been extensively studied in the framework of reconfiguration problems~\cite{IDHPSUU}, which arise when we wish to find a step-by-step transformation between two feasible solutions of a problem such that all intermediate solutions are also feasible.
	This reconfiguration framework has been applied to several well-studied combinatorial problems, including 
satisfiability~\cite{Kolaitis,MTY11}, 
independent set~\cite{BKW14,HearnDemaine2005,IDHPSUU,KMM12,MNRSS13},
set cover, matching~\cite{IDHPSUU}, 
shortest path~\cite{Bon12,Bon13,KMM11}, 
list edge-coloring~\cite{IKD09,IKZ11},
list $L(2,1)$-labeling~\cite{IKOZ14}, 
and so on.  

	In particular, the {\sc $\numk$-coloring reconfiguration} problem (i.e., {\sc $\numk$-list coloring reconfiguration} in which $\verlist(v) = \colorset$ holds for every vertex $v$) 
is one of the most well-studied reconfiguration problems, from the various viewpoints~\cite{BJLPP14,BB13,BC09,BM14,CHJ11,JKKPP14,Wro14}, as follows.

	Bonsma and Cereceda~\cite{BC09} proved that {\sc $k$-coloring reconfiguration} is PSPACE-complete for $k \ge 4$; 
they also proved that {\sc $\numk$-list coloring reconfiguration} is PSPACE-complete, even for bipartite planar graphs and $\numk = 4$.  
	On the other hand, Cereceda et al.~\cite{CHJ11} proved that {\sc $k$-coloring reconfiguration} is solvable for any graph in polynomial time for the case where $1 \le k \le 3$. 

	Then, some sufficient conditions have been proposed so that any pair of $\numk$-colorings of a graph has a desired transformation. 
	Cereceda~\cite{Cer07} gave a sufficient condition with respect to the number $\numk$ of colors:
if $\numk$ is at least the treewidth of a graph $G$ plus two, then there is a desired transformation between any pair of $\numk$-colorings of $G$;
the length of the transformation (i.e., the number of recoloring steps) is estimated by Bonamy and Bousquet~\cite{BB13}.
	Bonamy et al.~\cite{BJLPP14} gave a sufficient condition with respect to graph structures:
for example, chordal graphs and chordal bipartite graphs satisfy their sufficient condition. 
	
	Recently, Bonsma et al.~\cite{BM14} and Johnson et al.~\cite{JKKPP14} independently developed a fixed-parameter algorithm to solve {\sc $\numk$-coloring reconfiguration} when parameterized by $\numk + \ell$, where $\numk$ is the number of colors and $\ell$ is the number of recoloring steps.  
	In contrast, if the problem is parameterized only by $\ell$, then it is W[1]-hard~\cite{BM14} and does not admit a polynomial kernelization unless the polynomial hierarchy collapses~\cite{JKKPP14}. 

	In this way, even for the non-list version, only a few results are known from the viewpoint of polynomial-time solvability.
	Furthermore, as far as we known, no algorithmic result has been obtained for the list version. 
		
\subsection{Our contribution}

	In this paper, we study the {\sc $\numk$-list coloring reconfiguration} problem from the viewpoint of graph classes, especially pathwidth of graphs. 
(The definition of pathwidth will be given in Section~\ref{sec:pre}.)

	We first show that the problem remains PSPACE-complete even for graphs with pathwidth two. 
	In contrast, we give a polynomial-time algorithm to solve the problem for graphs with pathwidth one. 
	Thus, this paper gives precise analyses of the problem with respect to pathwidth. 

	Indeed, our reduction for the PSPACE-completeness proof constructs a bipartite series-parallel graph (whose treewidth is two), which is a bipartite planar graph. 
	We note that the problem of finding one $\numk$-list coloring of a given graph can be solved in polynomial time for bounded treewidth graphs (and hence for bounded pathwidth graphs)~\cite{JS97}.
	However, our proof shows that the reconfiguration variant is PSPACE-complete even if treewidth and pathwidth are two. 
	Furthermore, as a byproduct, our reduction can be extended for threshold  graphs.

\section{Preliminaries}
\label{sec:pre}

	We assume without loss of generality that graphs are simple and connected.
	Let $G=(V,E)$ be a graph with vertex set $V$ and edge set $E$;
we sometimes denote by $V(G)$ and $E(G)$ the vertex set and edge set of $G$, respectively. 
	For a vertex $v$ in $G$, we denote by $\degG(v)$ the degree of $v$ in $G$. 
	For a vertex subset $V^\prime \subseteq V$, we denote by $G[V^\prime]$ the subgraph of $G$ induced by $V^\prime$. 
\medskip

	We now define the notion of pathwidth~\cite{RS83}. 
	A {\em path-decomposition} of a graph $G$ is a sequence of subsets $X_i$ of vertices in $G$ such that 
    \begin{listing}{aaa}
       \item[(1)] $\bigcup_{i} X_i = V(G)$; 
       \item[(2)] for each $vw \in E(G)$, there is at least one subset $X_i$ with $v,w \in X_i$; and
       \item[(3)] for any three indices $p, q, r$ such that $p \le q \le r$, $X_p \cap X_r \subseteq X_q$.
    \end{listing}
	The {\em width} of a path-decomposition is defined as $\max_{i} |X_i|-1$, and the {\em pathwidth} of $G$ is the minimum $t$ such that $G$ has a path-decomposition of width $t$. 

	To develop our algorithm in Section~\ref{sec:algorithm}, it is important to notice that every connected graph of pathwidth one is a caterpillar~\cite{PT99}.
	A caterpillar will be defined in Section~\ref{sec:algorithm}, but an example can be found in \figurename~\ref{fig:caterpillar}. 

\medskip
	For a graph $G$ with a list $\verlist$, we define the {\em reconfiguration graph} $\Rgraph{G}{\verlist}$ as follows:
each node of $\Rgraph{G}{\verlist}$ corresponds to a $\numk$-list coloring of $G$, and two nodes of $\Rgraph{G}{\verlist}$ are joined by an edge if their corresponding $\numk$-list colorings $\mapf$ and $\mapf^\prime$ satisfy $|\{v \in V : \mapf(v) \neq \mapf^\prime(v) \}| = 1$, that is, $\mapf^\prime$ can be obtained from $\mapf$ by changing the color assignment of a single vertex $v$.
	We will refer to a {\em node} of $\Rgraph{G}{\verlist}$ in order to distinguish it from a vertex of $G$. 
	Since we have defined $\numk$-list colorings as the nodes of $\Rgraph{G}{\verlist}$, we use graph terms such as adjacency and path for $\numk$-list colorings. 
	For notational convenience, we sometimes identify a node of $\Rgraph{G}{\verlist}$ with its corresponding $\numk$-list coloring of $G$ if it is clear from the context. 


	Given a graph $G$ with a list $L$ and two $\numk$-list colorings $\mapf_{\ini}$ and $\mapf_{\tar}$ of $G$, the {\sc $\numk$-list coloring reconfiguration} problem asks whether the reconfiguration graph $\Rgraph{G}{\verlist}$ has a path between the two nodes $\mapf_{\ini}$ and $\mapf_{\tar}$.

\section{PSPACE-completeness}

	In this section, we first prove that the problem is PSPACE-complete even for graphs with pathwidth two. 
	Then, we show in Section~\ref{subsec:threshold} that the reduction can be extended to proving the PSPACE-completeness of threshold graphs. 
	\medskip
	
	A graph is {\em series-parallel} if it does not contain a subdivision of a complete graph $K_4$ on four vertices~\cite{BLS99}. 
	Note that series-parallel graphs may have super-constant pathwidth, although their treewidth can be bounded by two. 
	We give the following theorem. 
	\begin{theorem}
	\label{the:pspace_sp}
	The {\sc $\numk$-list coloring reconfiguration} problem is {\rm PSPACE}-complete even for bipartite series-parallel graphs of pathwidth two.
	\end{theorem}
	
	It is known that {\sc $\numk$-list coloring reconfiguration} is in PSPACE~\cite{BC09}.
	Therefore, as a proof of Theorem~\ref{the:pspace_sp}, we give a polynomial-time reduction from the {\sc shortest path rerouting} problem~\cite{Bon13} (defined in Section~\ref{subsec:shortest}) to our problem for bipartite series-parallel graphs of pathwidth two.

\subsection{Reconfiguration problem for shortest path}
\label{subsec:shortest}
	
	Let $\oriG$ be an unweighted graph, and let $\ters$ and $\tert$ be two vertices in $\oriG$. 
	We call a shortest path in $\oriG$ between $\ters$ and $\tert$ simply an {\em S-path} in $\oriG$. 
	Note that, since $\oriG$ is unweighted, an S-path indicates a path between $\ters$ and $\tert$ having the minimum number of edges. 
	We say that two S-paths $\Path$ and $\Path^\prime$ in $\oriG$ are {\em adjacent} if they differ in exactly one vertex, that is, $|V(\Path^\prime) \setminus V(\Path)| = 1$ and $|V(\Path) \setminus V(\Path^\prime)| = 1$ hold. 
	Given two S-paths $\Path_{\ini}$ and $\Path_{\tar}$ in $\oriG$, the {\sc shortest path rerouting} problem asks whether there exists a sequence $\seq{P}=\hseq{\Path_{\ini}, \Path_1, \ldots, \Path_{\ell}}$ of S-paths such that $\Path_{\ell} = \Path_{\tar}$, and $\Path_{i-1}$ and $\Path_{i}$ are adjacent for each $i = 1, 2, \ldots, \ell$. 
	This problem is known to be PSPACE-complete~\cite{Bon13}. 
	

	To construct our reduction, we introduce some terms.
	Let $\minlen$ be the number of edges of an S-path in a graph $\oriG$. 
	For two vertices $v$ and $w$ in $\oriG$, we denote by $\dist{v}{w}$ the number of edges of a shortest path in $\oriG$ between $v$ and $w$;
then $\dist{\ters}{\tert} = \minlen$. 
	For each $i \in \hset{0, 1, \ldots, \minlen}$, the \textit{layer $D_{\ters,i}$ from $\ters$} is defined to be the set of all vertices in $\oriG$ that are placed at distance $i$ from $\ters$, that is, 
$D_{\ters, i} = \hset{ v \in V(\oriG) : \dist{\ters}{v} = i }$. 
	Similarly, let $D_{\tert,j} = \hset{ v \in V(\oriG) : \dist{\tert}{v} = j }$ for each $j \in \hset{0, 1, \ldots, \minlen}$.
	Then, for each $i \in \hset{0, 1, \ldots, \minlen}$, the {\em layer $\layer{i}$ of $\oriG$} is defined as follows: 
$\layer{i} = D_{\ters,i} \cap D_{\tert, \minlen - i}$. 
	Notice that any S-path in $\oriG$ contains exactly one vertex from each layer $\layer{i}$, $0 \le i \le \minlen$. 
	Then, observe that two S-paths $\Path$ and $\Path^\prime$ in $\oriG$ are adjacent if and only if there exists exactly one index $j \in \hset{1, 2, \ldots, \minlen}$ such that $V(\Path) \cap \layer{j} \neq V(\Path^\prime) \cap \layer{j}$.
	Therefore, we may assume without loss of generality that $\oriG$ consists of only vertices in $\bigcup_{0 \le i \le \minlen} \layer{i}$, as illustrated in \figurename~\ref{fig:reductionSP}(a). 
	In the example of \figurename~\ref{fig:reductionSP}(a), all vertices in the same layer $\layer{i}$ are depicted by the same shape, that is, $\layer{0} = \{ \ters \}$, $\layer{1} = \{ v_{1,1}, v_{1,2} \}$, $\layer{2} = \{ v_{2,1}, v_{2,2}, v_{2,3} \}$, $\layer{3} = \{ v_{3,1}, v_{3,2} \}$ and $\layer{4} = \{ \tert \}$. 
	Note that both $\layer{0} = \{\ters\}$ and $\layer{\minlen} = \{\tert\}$ always hold.

\subsection{Reduction}
\label{subsec:reduction}



	Given an instance $(\oriG,\Path_{\ini},\Path_{\tar})$ of {\sc shortest path rerouting}, we construct the corresponding instance $(G,L,\mapf_{\ini},\mapf_{\tar})$ of {\sc $\numk$-list coloring reconfiguration}.
\medskip

\begin{figure}[t]
\begin{center}
\includegraphics[width=0.85\textwidth]{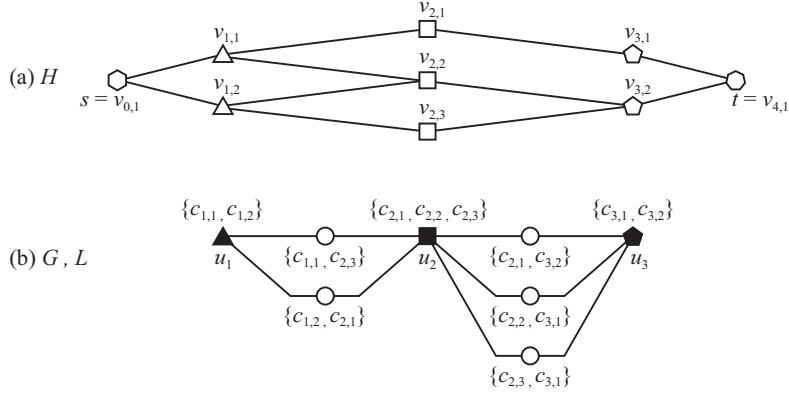}
\end{center}
\vspace{-2em}
\caption{(a) Graph $\oriG$ for {\sc shortest path rerouting}, and (b) the corresponding graph $G$ and its list $\verlist$ for {\sc $\numk$-list coloring reconfiguration}.}
\label{fig:reductionSP}
\vspace{-1em}
\end{figure}

\noindent
	{\bf Construction of $G$ and $\verlist$.}
	
	We first construct the corresponding graph $G$ with a list $L$. 
	For each $i \in \hset{1,2, \ldots, \minlen-1}$, let $\layer{i} = \{ v_{i,1}, v_{i,2}, \ldots, v_{i,q} \}$; and we introduce a vertex $u_i$, called a \textit{layer vertex}, to $G$. 
	The list of each layer vertex $u_i$ is defined as $\verlist(u_i) = \{ c_{i,1}, c_{i,2}, \ldots, c_{i,q} \}$, where each color $c_{i,j}$ in $\verlist(u_i)$ corresponds to the vertex $v_{i,j}$ in $\layer{i}$;
assigning color $c_{i,j}$ represents selecting the vertex $v_{i,j}$ as an S-path in $\oriG$. 
	We denote by $\Vlay$ the set of all layer vertices $u_1, u_2, \ldots, u_{\minlen-1}$ in $G$.
	In \figurename~\ref{fig:reductionSP}(b), each layer vertex $u_i$ is illustrated as a black vertex which is the same shape as the vertices in $\layer{i}$. 
	
	We then connect layer vertices in $G$ by {\em forbidden paths} of length two, as follows. 
	Let $v_{i,x} \in \layer{i}$ and $v_{i+1,y} \in \layer{i+1}$ be an arbitrary pair of vertices in $\oriG$ such that $v_{i,x} v_{i+1,y} \not\in E(\oriG)$.
	We introduce a vertex $w$ to $G$, and join $w$ and each of the two layer vertices $u_i$ and $u_{i+1}$ by an edge.
	The list $\verlist(w)$ of $w$ consists of two colors $c_{i,x}$ and $c_{i+1,y}$ which correspond to the vertices $v_{i,x}$ and $v_{i+1,y}$ in $\oriG$, respectively.
(See the vertices depicted by white circles in \figurename~\ref{fig:reductionSP}(b).)
	We call such a vertex $w$ in $G$ a $(v_{i,x}, v_{i+1,y})$\textit{-forbidden vertex} or simply a \textit{forbidden vertex}.
	Notice that there is no proper $\numk$-list coloring $\mapf$ such that $\mapf(u_i) = c_{i,x}$ and $\mapf(u_{i+1}) = c_{i+1,y}$;
otherwise there is no color in $\verlist(w)$ that can be assigned to $w$.
	This property ensures that every $\numk$-list coloring of $G$ corresponds to an S-path in $\oriG$. 
	Let $\Vfor$ be the set of all forbidden vertices, then $|\Vfor| = O(|E(\bar{\oriG})|)$ where $\bar{\oriG}$ is the complement graph of $\oriG$. 

	This completes the construction of $G$ and $\verlist$.
	Since $|V(G)| = O(d+|E(\bar{\oriG})|) = O(|V(\oriG)|^2)$ and $|E(G)| = O(|E(\bar{\oriG})|) = O(|V(\oriG)|^2)$, we can construct $G$ in polynomial time.
	Clearly, $G$ is a bipartite series-parallel graph of pathwidth two, whose bipartition consists of $\Vlay$ and $\Vfor$. 
\medskip

\noindent
	{\bf Construction of $\mapf_{\ini}$ and $\mapf_{\tar}$.}
	
	We now construct two $\numk$-list colorings $\mapf_{\ini}$ and $\mapf_{\tar}$ of $G$ which correspond to S-paths $\Path_{\ini}$ and $\Path_{\tar}$ in $\oriG$, respectively.
	For each $i \in \{ 1,2,\ldots, \minlen-1 \}$, let $v_{i,\ini}$ be the vertex in the layer $\layer{i}$ passed through by the S-path $\Path_{\ini}$;
then we let $\mapf_{\ini}(u_i) = c_{i,\ini}$ for each layer vertex $u_i \in \Vlay$. 
	For each $(v_{i,x}, v_{i+1,y})$-forbidden vertex $w \in \Vfor$, we choose an arbitrary color from $c_{i,x}$ and $c_{i+1,y}$ which is assigned to neither $u_i$ nor $u_{i+1}$. 
	We note that such an available color always exists, 
because 
$\Path_{\ini}$ has an edge between the two vertices corresponding to the colors $\mapf_{\ini}(u_i)$ and $\mapf_{\ini}(u_{i+1})$, and hence at least one of $\mapf_{\ini}(u_i) \neq c_{i,x}$ and $\mapf_{\ini}(u_{i+1}) \neq c_{i+1,y}$ holds for each $(v_{i,x}, v_{i+1,y})$-forbidden vertex. 
	Similarly, we construct $\mapf_{\tar}$. 
	This completes the construction of the corresponding instance $(G,L,\mapf_{\ini},\mapf_{\tar})$.
\medskip

\noindent
	{\bf Correctness of the reduction.}
	
	To show the correctness of this reduction, we give the following lemma.
	\begin{lemma} \label{lem:reduction}
	$(\oriG,\Path_{\ini},\Path_{\tar})$ is a yes-instance if and only if $(G,L,\mapf_{\ini},\mapf_{\tar})$ is a yes-instance.
%
	\end{lemma}
	\begin{proof}
	We first prove the if-part. 
	Suppose that the reconfiguration graph $\Rgraph{G}{\verlist}$ has a path between the two nodes $\mapf_{\ini}$ and $\mapf_{\tar}$. 
	We can classify the recoloring steps into the following two types:
(1) recoloring a layer vertex in $\Vlay$, and (2) recoloring a forbidden vertex in $\Vfor$. 
	Therefore, the path in $\Rgraph{G}{\verlist}$ can be divided into sub-paths, intermittently at each edge corresponding to a recoloring step of type~(1) above;
all edges in a sub-path correspond to recoloring steps of type~(2) above. 
	Therefore, all nodes in each sub-path correspond to the same S-path in $\oriG$.
	Furthermore, any two consecutive sub-paths correspond to two adjacent S-paths (that differ in only one vertex), because the sub-paths are divided by a recoloring step of type~(1).
	Thus, we can construct a sequence of adjacent S-paths that transforms $\Path_\ini$ into $\Path_\tar$.
	
	We then prove the only-if-part.
	Suppose that there exists a sequence $\seq{P}=\hseq{\Path_{\ini}, \Path_1, \ldots, \Path_{\ell}}$ of S-paths such that $\Path_{\ell} = \Path_{\tar}$, and $\Path_{j-1}$ and $\Path_{j}$ are adjacent for each $j = 1, 2, \ldots, \ell$.
	For two adjacent S-paths $\Path_{j-1}$ and $\Path_j$, $j \in \{1,2,\ldots, \ell\}$, assume that $\Path_{j}$ is obtained from $\Path_{j-1}$ by replacing $v_{i,x}\in V(\Path_{j-1})\cap \layer{i}$ with $v_{i,y}\in V(\Path_j) \cap \layer{i}$. 
(See \figurename~\ref{fig:recoloring}(a).)
	This rerouting step from $v_{i,x}$ to $v_{i,y}$ corresponds to recoloring the layer vertex $u_i \in \Vlay$ from the color $c_{i,x}$ to $c_{i,y}$.
(See \figurename~\ref{fig:recoloring}(b).)
	To do so, if there is a forbidden vertex $w \in \Vfor$ which is adjacent with $u_i$ and receives the color $c_{i,y}$, we first need to recolor $w$ from $c_{i,y}$ to another color $c_{q,z} \in \verlist(w) \setminus \{ c_{i,y} \}$, where $q \in \{i-1, i+1\}$.
 (Note that, since $\Path_{j-1}$ corresponds to a feasible $\numk$-list coloring of $G$, no forbidden vertex adjacent with $u_i$ receives the color $c_{i,x}$.)
	Since $w$ is a forbidden vertex, we know that $|\verlist(w)| = 2$.
	Let $v_{q,a}$ be the vertex in the layer $\layer{q}$ which is adjacent with $v_{i,x}$ and $v_{i,y}$ in $\Path_{j-1}$ and $\Path_j$, respectively;
and hence the layer vertex $u_{q}$ receives the color $c_{q,a}$.
(Figure~\ref{fig:recoloring} illustrates the case where $q=i-1$.)
	Then, since each forbidden vertex is placed only for a pair of vertices in $\oriG$ which is {\em not} joined by an edge, we observe that the color $c_{q,z} \in \verlist(w) \setminus \{ c_{i,y} \}$ is different from $c_{q,a}$. 
 	Therefore, we can recolor $w$ from $c_{i,y}$ to another color $c_{q,z}$ without recoloring any other vertices.
	Since $\Vfor$ forms an independent set of $G$, we can apply this recoloring step to all forbidden vertices independently. 
	Now, any neighbor of $u_i$ is colored with neither $c_{i,x}$ nor $c_{i,y}$, and hence we can recolor $u_i$ from $c_{i,x}$ to $c_{i,y}$.
	Thus, $\Rgraph{G}{\verlist}$ has a path between $\mapf_\ini$ and $\mapf_\tar$ which corresponds to $\seq{P}$.
	\qed
	\end{proof}
	
\begin{figure}[t]
\begin{center}
\includegraphics[width=\textwidth]{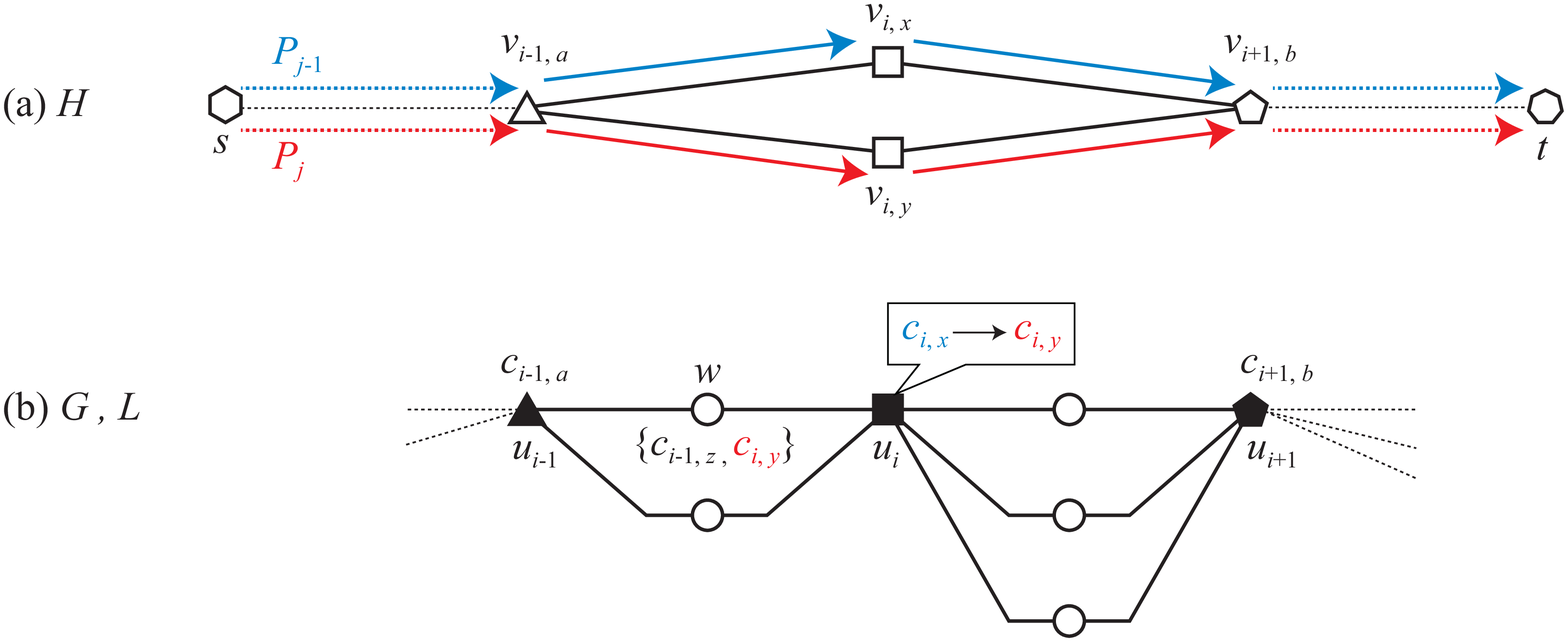}
\end{center}
\vspace{-2em}
\caption{(a) Two adjacent S-paths $\Path_{j-1}$ (blue) and $\Path_j$ (red), and (b) the corresponding recoloring steps.}
\label{fig:recoloring}
\vspace{-1em}
\end{figure}

\subsection{Threshold graphs}
\label{subsec:threshold}
	
	In this subsection, we extend our reduction in Section~\ref{subsec:reduction} to threshold graphs. 
	A graph $G$ is {\em threshold} if there exist a real number $s$ and a mapping $\omomi: V(G) \to \RRR$ such that $xy \in E(G)$ if and only if $\omomi(x) + \omomi(y) \ge s$~\cite{BLS99}, where $\RRR$ is the set of all real numbers. 
	\begin{theorem}
	\label{the:pspace_thre}
	The {\sc $\numk$-list coloring reconfiguration} problem is {\rm PSPACE}-complete even for threshold graphs.
	\end{theorem}
	\begin{proof}
	We modify the graph $G$ constructed in Section~\ref{subsec:reduction}, as follows: 
join all pairs of vertices in $\Vlay$, and join each vertex in $\Vfor$ with all vertices in $\Vlay$. 
	Let $G^\prime$ be the resulting graph, then $G[\Vlay]$ forms a clique and $G[\Vfor]$ forms an independent set. 
	Notice that $G^\prime$ is a threshold graph;
set the threshold $s = 1$, and the mapping $\omomi(v) = 1$ for each layer vertex $v \in \Vlay$ and $\omomi(u) = 0$ for each forbidden vertex $u \in \Vfor$.

	Consider any pair of vertices $u$ and $v$ such that $uv \in E(G^\prime) \setminus E(G)$, that is, they are joined by a new edge for constructing $G^\prime$ from $G$. 
	Then, by the construction in Section~\ref{subsec:reduction}, the two lists $\verlist(u)$ and $\verlist(v)$ contain no color in common. 
	Therefore, adding new edges to $G$ does not affect the existence of $\numk$-list coloring:
more formally, any $\numk$-list coloring of $G$ is a $\numk$-list coloring of $G^\prime$, and vice versa. 
	Thus, by Lemma~\ref{lem:reduction}, an instance $(\oriG,\Path_{\ini},\Path_{\tar})$ of {\sc shortest path rerouting} is a yes-instance if and only if $(G^\prime,L,\mapf_{\ini},\mapf_{\tar})$ is a yes-instance.
	\qed
	\end{proof}

	\section{Algorithm for Graphs with Pathwidth One}
	\label{sec:algorithm}

	In contrast to Theorem~\ref{the:pspace_sp}, we give the following theorem in this section. 
	\begin{theorem}
	\label{the:caterpillar}
	The {\sc $\numk$-list coloring reconfiguration} problem can be solved in polynomial time for graphs with pathwidth one.
	\end{theorem}

	As a proof of Theorem~\ref{the:caterpillar}, we give such an algorithm.
	However, since every connected graph of pathwidth one is a caterpillar~\cite{PT99}, it suffices to develop a polynomial-time algorithm 
for caterpillars.

	A \textit{caterpillar} $G$ is a tree whose vertex set $V(G)$ can be partitioned into two subset $\Vspine$ and $\Vleaf$ such that $G[\Vspine]$ forms a path 
and each vertex in $\Vleaf$ is incident to exactly one vertex in $\Vspine$.
	We may assume without loss of generality that the two endpoints of the path $G[\Vspine]$ are of degree one in the whole graph $G$. 
(See $v_1$ and $v_{10}$ in \figurename~\ref{fig:caterpillar}.)
	We call each vertex in $\Vspine$ a \textit{spine vertex} of $G$, and each vertex in $\Vleaf$ a \textit{leaf} of $G$.

\begin{figure}[b]
\begin{center}
\includegraphics[width=0.5\textwidth]{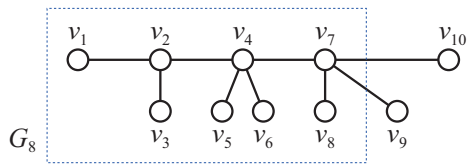}
\end{center}
\vspace{-2em}
\caption{A caterpillar $G$ and its vertex ordering, where the subgraph surrounded by a dotted rectangle corresponds to $G_8$.}
\label{fig:caterpillar}
\end{figure}

	We assume that all vertices in $G$ are ordered as $v_1, v_2, \ldots, v_n$ by the breadth-first search starting with the endpoint (degree-$1$ vertex) of the spine path $G[\Vspine]$ with the priority to leaves;
that is, when we visit a spine vertex $v$, we first visit all leaves of $v$ and then visit the unvisited spine vertex.
(See \figurename~\ref{fig:caterpillar} for example.)
	For each index $i \in \hset{1,2,\ldots, n}$, we let $V_i = \hset{v_1, v_2, \ldots, v_i}$ and $G_i = G[V_i]$.
	Then, clearly $G_n=G$.	
	For each index $i \in \hset{1,2,\ldots, n}$, let $\spine{i}$ be the latest spine vertex in $V_i$, that is, $\spine{i}=v_i$ if $v_i$ is a spine vertex, otherwise $\spine{i}$ is the unique neighbor of $v_i$.
	Then, $v_{i}$ is adjacent with only the spine vertex $\spine{i-1}$ in $G_{i}$ for each $i \in \{2,3,\ldots, n\}$.  
\medskip

	We can restrict the length of each list without loss of generality, as follows.
	Note that the following lemma holds for any graph. 
	\begin{lemma} \label{lem:transform}
	For an instance $(G^\prime, \verlist^\prime, \mapf^\prime_{\ini}, \mapf^\prime_{\tar})$, one can obtain another instance $(G, \verlist, \mapf_{\ini}, \mapf_{\tar})$ in polynomial time such that $2 \le |\verlist(v)| \le \degG(v)+1$ for each vertex $v \in V(G)$, and $(G^\prime, \verlist^\prime, \mapf^\prime_{\ini}, \mapf^\prime_{\tar})$ is a yes-instance if and only if $(G, \verlist, \mapf_{\ini}, \mapf_{\tar})$ is a yes-instance. 
	\end{lemma}
	\begin{proof}
	If $|\verlist(v)| = 1$ for a vertex $v \in V(G^\prime)$, then any $\numk$-list coloring $\mapf^\prime$ of $G^\prime$ assigns the same color $c \in \verlist(v)$ to $v$. 
	Therefore, $\mapf^\prime$ never assigns $c$ to any neighbor $u$ of $v$. 
	We can thus delete $v$ from $G^\prime$ and set $\verlist(u) := \verlist(u) \setminus \{ c\}$ for all neighbors $u$ of $v$ in $G^\prime$. 
	Clearly, this modification does not affect the reconfigurability (i.e., the existence or non-existence of a path in the reconfiguration graph).

	If $|\verlist(v)| \ge \degG(v) + 2$ for a vertex $v \in V(G^\prime)$, we simply delete $v$ from $G^\prime$ without any modification of lists;
let $G$ be the resulting graph.
	Let $\mapf$ be any $\numk$-list coloring of $G$, and consider a recoloring step for a neighbor $u$ of $v$ from the current color $c = \mapf(u)$ to another color $c^{\prime}$. 
	If $c^\prime$ is not assigned to $v$ in $G^\prime$, we can directly recolor $u$ from $c$ to $c^{\prime}$.
	Thus, suppose that $c^{\prime}$ is assigned to $v$ in $G^\prime$. 
	Then, since $|\verlist(v)| \ge \degG(v) +2$, there is at least one color $c^* \in \verlist(v)$ which is not $c^{\prime}$ and is not assigned to any of $\degG(v)$ neighbors of $v$ by $\mapf$.
	Therefore, we first recolor $v$ from $c^{\prime}$ to $c^*$, and then recolor $u$ from $c$ to $c^{\prime}$. 
	In this way, any recoloring step in $G$ can be simulated in $G^\prime$, and hence the modification does not affect the reconfigurability. 

	Thus, we can obtain an instance such that $2 \le |\verlist(v)| \le \degG(v)+1$ holds for each vertex $v$ without affecting the reconfigurability. 
	Clearly, the modified instance can be constructed in polynomial time.
	\qed
	\end{proof}
	
	Therefore, in the remainder of this section, we assume that $G$ is a (connected) caterpillar and $2 \le |\verlist(v)| \le \degG(v)+1$ holds for every vertex $v \in V(G)$.
	In particular, $|\verlist(v)| = 2$ for every leaf $v$ of $G$.

\subsection{Idea and definitions}

	The main idea of our algorithm is to extend techniques developed for {\sc shortest path rerouting}~\cite{Bon12}, and apply them to {\sc $\numk$-list coloring reconfiguration} for caterpillars. 
	Our algorithm employs a dynamic programming method based on the vertex ordering $v_1, v_2, \ldots, v_n$ of $G$. 

	For each $i \in \{1,2, \ldots, n\}$, let $\Rgraph{G_i}{L}$ be the reconfiguration graph for the subgraph $G_i$ and the list $L$. 
	Then, $\Rgraph{G_i}{L}$ contains all $\numk$-list colorings of $G_i$ as its nodes.
	Our algorithm efficiently constructs $\Rgraph{G_i}{L}$ for each $i = 1,2, \ldots, n$, in this order.
	However, of course, the number of nodes in $\Rgraph{G_i}{L}$ cannot be bounded by a polynomial size in general.
	We thus use the property that the vertex $v_{i+1}$ (will be added to $G_i$) is adjacent with only the spine vertex $\spine{i}$ in $G_{i+1}$;
and we ``encode'' the reconfiguration graph $\Rgraph{G_i}{L}$ into a polynomial size with keeping the information of (1) the color assigned to $\spine{i}$, and (2) the connectivity of nodes in $\Rgraph{G_i}{L}$. 
	
	Before explaining the encoding methods, we first note that it suffices to focus on only one connected component in $\Rgraph{G_i}{L}$ which contains the restriction of $\mapf_{\ini}$, where the {\em restriction} of a $\numk$-list coloring $\mapf$ of a graph $G$ to a subgraph $G^\prime$ is a $\numk$-list coloring $\mapi$ of $G^\prime$ such that $\mapi(v) = \mapf(v)$ hold for all vertices $v \in V(G^\prime)$. 
	For notational convenience, we denote by $\mapf[V_i]$ the restriction of a $\numk$-list coloring $\mapf$ of a caterpillar $G$ to its subgraph $G_i$. 
	Then, we have the following lemma.
	\begin{lemma} \label{lem:restrict}
	Let $\mapg$ be a $\numk$-list coloring of $G_i$ such that $\mapf_{\ini}[V_i]$ and $\mapg$ are contained in the same connected component in $\Rgraph{G_i}{L}$. 
	Then, for each $j \in \{1,2, \ldots, i-1\}$, $\mapf_{\ini}[V_j]$ and $\mapg[V_j]$ are contained in the same connected component in $\Rgraph{G_j}{L}$. 
	\end{lemma}
	\begin{proof}
	Assume that $\mapf_{\ini}[V_i]$ and $\mapg$ are contained in the same connected component in $\Rgraph{G_i}{L}$.
	Then, there exists a path in $\Rgraph{G_i}{L}$ between $\mapf_{\ini}[V_i]$ and $\mapg$.
	We contract all edges in the path that correspond to recoloring vertices in $V_i \setminus V_j$.
	Since each edge in the resulting path corresponds to recoloring only one vertex in $V_j$, the resulting path must be contained as a path in $\Rgraph{G_j}{L}$ between $\mapf_{\ini}[V_j]$ and $\mapg[V_j]$.
	Therefore $\mapf_{\ini}[V_j]$ and $\mapg[V_j]$ are contained in the same connected component in $\Rgraph{G_j}{L}$, and hence the lemma follows. 
	\qed
	\end{proof}
	
	From now on, we thus focus on only the connected component of $\Rgraph{G_i}{L}$ which contains $\mapf_{\ini}[V_i]$.
	Since the list is fixed to be $L$ in the remainder of this section, we simply denote by $\Ri{i}$ the reconfiguration graph $\Rgraph{G_i}{L}$, and by $\Rizero{i}$ the connected component of $\Ri{i} = \Rgraph{G_i}{L}$ containing $\mapf_{\ini}[V_i]$.
\medskip


		
	

\noindent
	{\bf Encoding graph.}
	
	We now partition the nodes of $\Rizero{i}$ into several subsets with respect to (1) the color assigned to $\spine{i}$, and (2) the connectivity of nodes in $\Rizero{i}$. 
	For two nodes $\mapi$ and $\mapip$ of $\Rizero{i}$ with $\mapi(\spine{i}) = \mapip(\spine{i})$, we write $\mapi \eqrel{\spine{i}} \mapip$ if $\mapi$ can be reconfigured into $\mapip$ without recoloring the color assigned to the vertex $\spine{i}$, that is, 
$\Rizero{i}$ has a path $\hseq{\mapi_1, \mapi_2, \ldots, \mapi_\ell}$ such that $\mapi_1 = \mapi$, $\mapi_{\ell} = \mapip$, and $\mapi_j(\spine{i}) = \mapi(\spine{i}) = \mapip(\spine{i})$ holds for every $j \in \hset{1, 2, \ldots, \ell}$.
	Since the adjacency relation on $\numk$-list colorings is symmetric (i.e., $\Ri{i}$ is an undirected graph), it is easy to see that $\eqrel{\spine{i}}$ is an equivalence relation. 
	Thus, the node set of $\Rizero{i}$ can be uniquely partitioned by the relation $\eqrel{\spine{i}}$.
	We denote by $\parF{i}$ the partition of the node set of $\Rizero{i}$ into equivalence classes with respect to $\eqrel{\spine{i}}$.


\begin{figure}[t]
\begin{center}
\includegraphics[width=0.9\textwidth]{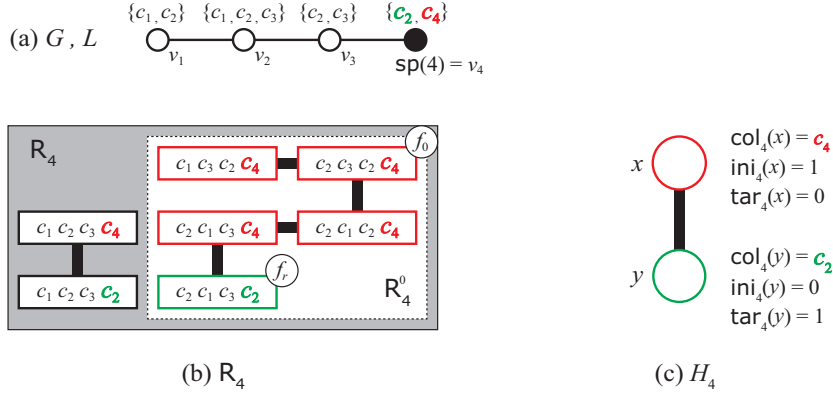}
\end{center}
\vspace{-1em}
\caption{(a) A caterpillar $G = G_4$, (b) the reconfiguration graph $\Ri{4}$ consisting of all $\numk$-list colorings of $G$, and (c) the encoding graph $\Enc{4}$ of $\Rizero{4}$ consisting of two e-nodes $x$ and $y$, where each $\numk$-list coloring in (b) is represented as the sequence of colors assigned to the vertices in $G$ from left to right.}
\label{fig:encoding}
\end{figure}
	
	We finally define our dynamic programming table.
	For each subgraph $G_i$, $i \in \{1, 2, \ldots, n\}$, our algorithm keeps track of four information $(\Enc{i}, \labcf{i}, \labinif{i}, \labtarf{i})$, defined as follows.
	\begin{itemize}
		\item The \textit{encoding graph} $\Enc{i}$ of $\Rizero{i}$ which can be obtained from $\Rizero{i}$ by contracting each node set in $\parF{i}$ into a single node.
		(See \figurename~\ref{fig:encoding} as an example.)
				We will refer to an {\em e-node} of $\Enc{i}$ in order to distinguish it from a node of $\Rizero{i}$.
(Thus, each node refers to a $\numk$-list coloring of $G_i$, and each e-node refers to a set of $\numk$-list colorings of $G_i$.) 
				For each e-node $x \in V(\Enc{i})$, we denote by $\mapset{i}{x}$ the set of all nodes in $\Rizero{i}$ that were contracted into $x$.
				(Note that we do not compute $\mapset{i}{x}$, but use it only for definitions and proofs.)
\smallskip

		\item The color $\labc{i}{x} \in \verlist(\spine{i})$ for each e-node $x \in V(\Enc{i})$, which is assigned to $\spine{i}$ in common by the nodes (i.e., $\numk$-list colorings of $G_i$) in $\mapset{i}{x}$.
\smallskip

		\item The label $\labini{i}{x} \in \{0, 1\}$ for each e-node $x \in V(\Enc{i})$, such that $\labini{i}{x}=1$ if $\mapf_\ini[V_i]\in \mapset{i}{x}$, otherwise $\labini{i}{x} = 0$.
\smallskip

		\item The label $\labtar{i}{x} \in \{0, 1\}$ for each e-node $x \in V(\Enc{i})$, such that $\labtar{i}{x}=1$ if $\mapf_\tar[V_i]\in \mapset{i}{x}$, otherwise $\labtar{i}{x} = 0$.
	\end{itemize}

	To prove Theorem~\ref{the:caterpillar}, we give a polynomial-time algorithm which computes $(\Enc{i}, \labcf{i}, \labinif{i}, \labtarf{i})$ for each subgraph $G_i$, $i \in \{1, 2, \ldots, n\}$, by means of dynamic programming.
	Then, the problem can be solved as in the following lemma.	
	\begin{lemma} \label{lem:hantei}
	$(G, \verlist, \mapf_{\ini}, \mapf_{\tar})$ is a yes-instance if and only if the encoding graph $\Enc{n}$ contains a node $x$ such that $\labtar{n}{x} = 1$.
	\end{lemma}
	\begin{proof}
	Since $\Enc{n}$ contains a node $x$ such that $\labtar{n}{x} = 1$, we have $\mapf_\tar[V_n] = \mapf_\tar \in \mapset{n}{x}$.
	Recall that $\Enc{n}$ is the encoding graph of $\Rizero{n}$ which contains the $\numk$-list coloring $\mapf_{\ini}$ of $G$ as a node.
	Since $\mapset{n}{x} \subseteq V(\Rizero{n})$, the lemma clearly follows.	
	\qed
	\end{proof}

\begin{figure}[t]
\begin{center}
\includegraphics[width=\textwidth]{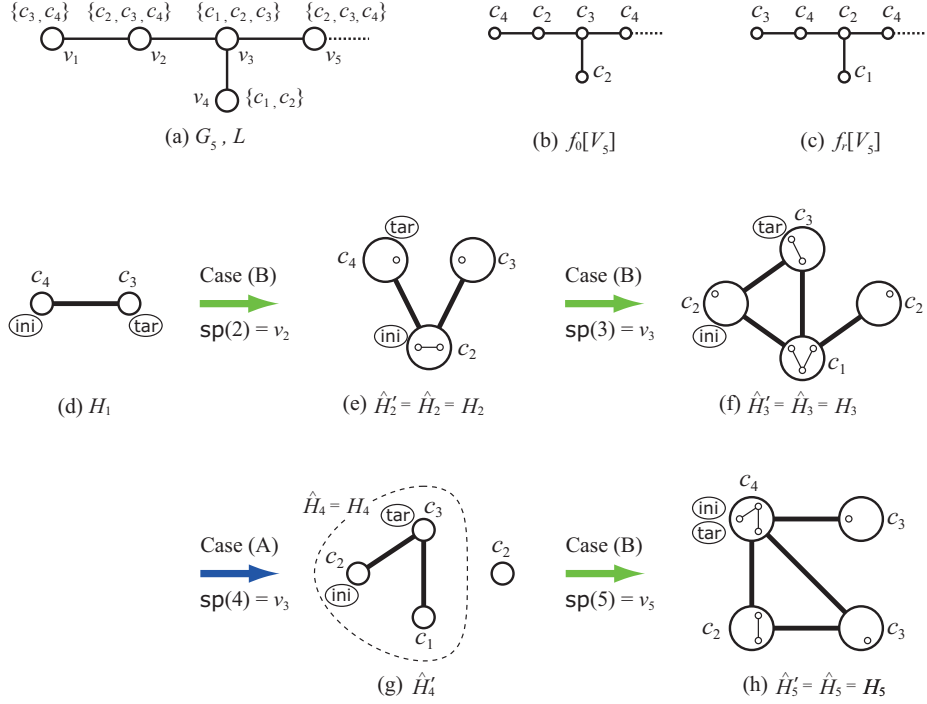}
\end{center}
\vspace{-2em}
\caption{Application of our algorithm.
In (d)--(h), $\labc{i}{x} \in \verlist(\spine{i})$ is attached to each e-node $x$, and the e-nodes $x$ with $\labini{i}{x} = 1$ and $\labtar{i}{x} = 1$ have the labels ``{\sf ini}'' and ``{\sf tar},'' respectively.
Furthermore, in (e), (f) and (h), the small graph contained in each e-node $x$ of $\Enc{i}$ represents the subgraph of $\Enc{i-1}$ induced by $\nset{x}$.}
\label{fig:algorithm}
\end{figure}

	\subsection{Algorithm}
	\label{subsec:alg}

	As the initialization, we first consider the case where $i = 1$, that is, 
we compute $(\Enc{1}, \labcf{1}, \labinif{1}, \labtarf{1})$. 
(See \figurename~\ref{fig:algorithm}(d) as an example.)
	Note that $G_1$ consists of a single vertex $v_1$, and recall that $v_1$ is a spine vertex of degree one. 
	By Lemma~\ref{lem:transform} we then have $|\verlist(v_1)| = 2$. 
	Therefore, the reconfiguration graph $\Ri{1}$ is a complete graph on $|\verlist(v_1)| = 2$ nodes such that each node corresponds to a $\numk$-list coloring of $G_1$ which assigns a distinct color to the vertex $\spine{1} = v_1$.
	Since $\Ri{1}$ is complete and contains the node $\mapf_{\ini}[V_1]$, we have $\Rizero{1} = \Ri{1}$. 
	Furthermore, $\Enc{1} = \Rizero{1}$ since all nodes in $\Rizero{1}$ assign distinct colors in $\verlist(v_1)$ to $\spine{1} = v_1$. 
	Then, for each e-node $x$ of $\Enc{1}$ corresponding to the set consisting of a single $\numk$-list coloring $\mapi$ of $G_1$, we set 
	\begin{eqnarray*}
	\labc{1}{x} &=& \mapi(v_1); \\
	\labini{1}{x} &=& \left\{
		\begin{array}{ll}
		1 & ~~~\mbox{if $\mapi(v_1) = \mapf_{\ini}(v_1)$},\\
		0 & ~~~\mbox{otherwise};
		\end{array} \right. \\
	\labtar{1}{x} &=& \left\{
		\begin{array}{ll}
		1 & ~~~\mbox{if $\mapi(v_1) = \mapf_{\tar}(v_1)$}, \\
		0 & ~~~\mbox{otherwise}.
		\end{array} \right.
	\end{eqnarray*}


\medskip

\begin{figure}[t]
\begin{center}
\includegraphics[width=0.8\textwidth]{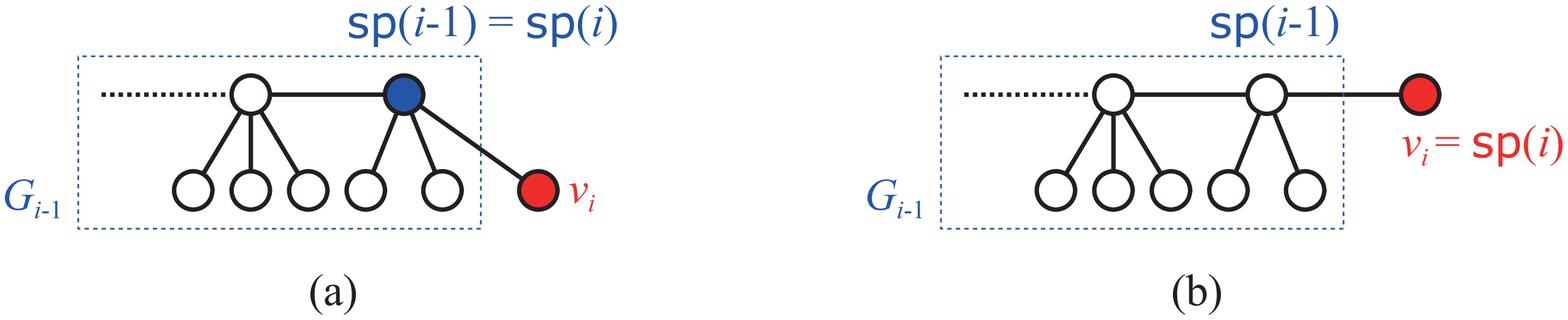}
\end{center}
\vspace{-2em}
\caption{The graph $G_i$ for (a) $v_i \in \Vleaf$ and (b) $v_i \in \Vspine$.}
\label{fig:update}
\vspace{-1em}
\end{figure}

	For $i \ge 2$, suppose that we have already computed $(\Enc{i-1}, \labcf{i-1}, \labinif{i-1}, \labtarf{i-1})$.
	Then, we compute $(\Enc{i}, \labcf{i}, \labinif{i}, \labtarf{i})$, as follows.
\medskip

\noindent
	{\bf Case (A): $v_i$ is a leaf in $\Vleaf$.} 
(See Figs.~\ref{fig:algorithm}(g) and \ref{fig:update}(a).)
	
	By Lemma~\ref{lem:transform} we have $|\verlist(v_i)| = 2$ in this case; let $\verlist(v_i) = \{ c_1, c_2 \}$. 
	Recall that $v_i$ is adjacent with only the spine vertex $\spine{i-1}$ in $G_i$.  
	Furthermore, $\spine{i} = \spine{i-1}$ in this case. 

	Let $\Encsub{i-1}{c_1}$ be the subgraph of $\Enc{i-1}$ obtained by deleting all e-nodes $y$ in $\Enc{i-1}$ with $\labc{i-1}{y} = c_1$. 
	Then, $\Encsub{i-1}{c_1}$ encodes all nodes of $\Rizero{i-1}$ that do not assign the color $c_1$ to $\spine{i-1}$.
	Thus, we can extend each $\numk$-list coloring $\mapj$ of $G_{i-1}$ encoded in $\Encsub{i-1}{c_1}$ to a $\numk$-list coloring $\mapi$ of $G_i$ such that $\mapi(v_i) = c_1$ and $\mapi(v) = \mapj(v)$ for all vertices $v \in V_{i-1}$. 
	Similarly, let $\Encsub{i-1}{c_2}$ be the subgraph of $\Enc{i-1}$ obtained by deleting all e-nodes $z$ in $\Enc{i-1}$ with $\labc{i-1}{z} = c_2$. 

	We define an encoding graph $\Encp{i}$ as 
$V(\Encp{i}) = V(\Encsub{i-1}{c_1}) \cup V(\Encsub{i-1}{c_2})$ and 
$E(\Encp{i}) = E(\Encsub{i-1}{c_1}) \cup E(\Encsub{i-1}{c_2})$;
and let $\Enct{i}$ be the connected component of $\Encp{i}$ that contains the e-node $x$ such that $\labini{i-1}{x} = 1$.
	For each e-node $x$ in $\Enct{i}$, let $\labct{i}{x} = \labc{i-1}{x}$, $\labinit{i}{x} = \labini{i-1}{x}$ and $\labtart{i}{x} = \labtar{i-1}{x}$.
	Then, we have the following lemma. 
	\begin{lemma} \label{lem:leaf}
	For a leaf $v_i \in \Vleaf$, $(\Enc{i}, \labcf{i}, \labinif{i}, \labtarf{i}) = (\Enct{i}, \labcft{i}, \labinift{i}, \labtarft{i})$.
	\end{lemma}
\medskip

\noindent
	{\bf Case (B): $v_i$ is a spine vertex in $\Vspine$.} 
(See Figs.~\ref{fig:algorithm}(e), (f), (h) and \ref{fig:update}(b).)

	In this case, notice that $\spine{i} = v_i$ in $G_i$, and hence we need to update $\labcf{i}$ according to the color assigned to $v_i$.
	
	We first define an encoding graph $\Encp{i}$, as follows. 
	For a color $c \in \verlist(v_i)$, let $\Encsub{i-1}{c}$ be the subgraph of $\Enc{i-1}$ obtained by deleting all e-nodes $y$ in $\Enc{i-1}$ with $\labc{i-1}{y} = c$.
	For each connected component in $\Encsub{i-1}{c}$, we add a new e-node $x$ to $\Encp{i}$ such that $\labct{i}{x} = c$;
we denote by $\nset{x}$ the set of all e-nodes in $\Encsub{i-1}{c}$ that correspond to $x$. 
	We apply this operation to all colors in $\verlist(v_i)$.
	We then add edges to $\Encp{i}$: two e-nodes $x$ and $y$ in $\Encp{i}$ are joined by an edge if and only if $\nset{x} \cap \nset{y} \neq \emptyset$. 
	
	We now define $\labinit{i}{x}$ and $\labtart{i}{x}$ for each e-node $x$ in $\Encp{i}$, as follows: 
	\[
		\labinit{i}{x} = \left\{
				\begin{array}{ll}
				1 & ~~~\mbox{if $\labct{i}{x} = \mapf_{\ini}(v_i)$ and } \\
				   & ~~~\mbox{$\nset{x}$ contains an e-node $y$ with $\labini{i-1}{y} = 1$}; \\
				0 & ~~~\mbox{otherwise}, 
				\end{array} \right.
	\]
and 
	\[
		\labtart{i}{x} = \left\{
				\begin{array}{ll}
				1 & ~~~\mbox{if $\labct{i}{x} = \mapf_{\tar}(v_i)$ and } \\
				   & ~~~\mbox{$\nset{x}$ contains an e-node $y$ with $\labtar{i-1}{y} = 1$}; \\
				0 & ~~~\mbox{otherwise}. 
				\end{array} \right.
	\]

	Let $\Enct{i}$ be the connected component of $\Encp{i}$ that contains the e-node $x$ such that $\labinit{i}{x} = 1$.
	Then, we have the following lemma. 
	\begin{lemma} \label{lem:spine}
	For a spine vertex $v_i \in \Vspine$, $(\Enc{i}, \labcf{i}, \labinif{i}, \labtarf{i}) = (\Enct{i}, \labcft{i}, \labinift{i}, \labtarft{i})$.
	\end{lemma}

	\subsection{Running time}

	We now estimate the running time of our algorithm in Section~\ref{subsec:alg}. 
	The following is the key lemma for the estimation. 
	\begin{lemma} \label{lem:size}
	For each index $i \in \{1,2,\ldots, n\}$, 
	\[
		|V(\Enc{i})| \le \left\{
				\begin{array}{ll}
				2 & ~~~\mbox{if $i = 1$}; \\
				|V(\Enc{i-1})| + \degG(v_i) & ~~~\mbox{otherwise}.
				\end{array} \right.
	\]
	In particular, $|V(\Enc{n})| = O(n)$, where $n$ is the number of vertices in $G$. 
	\end{lemma}
	
	By Lemma~\ref{lem:size} each encoding graph $\Enc{i}$ is of size $O(n)$ for each $i \in \{1,2,\ldots, n\}$. 
	Therefore, our algorithm clearly runs in polynomial time. 
	
	This completes the proof of Theorem~\ref{the:caterpillar}.
\qed

\section{Concluding Remarks}

	In this paper, we gave precise analyses of the {\sc $\numk$-list coloring reconfiguration} problem with respect to pathwidth:
the problem is solvable in polynomial time for graphs with pathwidth one, while it is PSPACE-complete for graphs with pathwidth two. 

	Very recently, Wrochna~\cite{Wro14} gave another proof for the PSPACE-completeness of {\sc $\numk$-list coloring reconfiguration} for graphs with pathwidth two. 
	His reduction is constructed from a PSPACE-complete problem, called {\sc $H$-Word Reachability}.

	\subsection*{Acknowledgments}
	We are grateful to Daichi Fukase and Yuma Tamura for fruitful discussions with them. 
	This work is partially supported by JSPS KAKENHI Grant Numbers 25106504 and 25330003.
	

\bibliographystyle{abbrv}

\input{appendix.tex}

\end{document}

%% file: newcom.tex
\newcommand{\mapipp}{\mapi^{\prime \prime}}		
\newcommand{\extend}[3]{\mapset{#1}{#2}\oplus #3}

\newcommand{\xhat}{\hat{x}}
\newcommand{\yhat}{\hat{y}}
\newcommand{\zhat}{\hat{z}}

\newcommand{\comp}[1]{{\sf cc}(#1)}
\newcommand{\stree}[1]{T_{#1}}
\newcommand{\Encst}[1]{\hat{H}_{#1}^T}
\newcommand{\Encsubst}[2]{H_{#1}^{T, #2}}
\newcommand{\Xic}[2]{X_{#1}(#2)}

\newcommand{\degG}{d}

\newcommand{\mapf}{f}
\newcommand{\verlist}{L}
\newcommand{\colorset}{C}
\newcommand{\numk}{k}
\newcommand{\cind}[1]{#1}

\newcommand{\oriG}{H}
\newcommand{\ters}{s}
\newcommand{\tert}{t}
\newcommand{\Path}{P}
\newcommand{\minlen}{d}
\newcommand{\dist}[2]{{\sf dist}(#1, #2)}

\newcommand{\layer}[1]{D_{#1}}

\newcommand{\Vlay}{U_{\sf layer}}
\newcommand{\Vfor}{U_{\sf forbid}}

\newcommand{\omomi}{\omega}
\newcommand{\RRR}{\mathbb{R}}
	
\newcommand{\Vspine}{V_{{\rm S}}}
\newcommand{\Vleaf}{V_{{\rm L}}}

\newcommand{\mapg}{g}

\newcommand{\Ri}[1]{{\sf R}_{#1}}
\newcommand{\Rizero}[1]{\Ri{#1}^{\ini}}
\newcommand{\mapi}{g}
\newcommand{\mapip}{\mapi^\prime}

\newcommand{\mapj}{h}
\newcommand{\nset}[1]{{\sf EN}(#1)}
	
\newcommand{\eqrel}[1]{\sim_{#1}}
\newcommand{\spine}[1]{{\sf sp}(#1)}
\newcommand{\parF}[1]{\mathcal{G}_{#1}^{\ini}}

\newcommand{\Enc}[1]{H_{#1}}
\newcommand{\Encsub}[2]{\Enc{#1}^{#2}}
\newcommand{\Enct}[1]{\hat{H}_{#1}}
\newcommand{\Encp}[1]{\hat{H}_{#1}^\prime}

\newcommand{\mapset}[2]{\Phi_{#1} (#2)}

\newcommand{\labcf}[1]{{\sf col}_{#1}}
\newcommand{\labinif}[1]{{\sf ini}_{#1}}
\newcommand{\labtarf}[1]{{\sf tar}_{#1}}

\newcommand{\labc}[2]{{\sf col}_{#1}(#2)}
\newcommand{\labini}[2]{{\sf ini}_{#1}(#2)}
\newcommand{\labtar}[2]{{\sf tar}_{#1}(#2)}

\newcommand{\labcft}[1]{\hat{{\sf col}}_{#1}}
\newcommand{\labinift}[1]{\hat{{\sf ini}}_{#1}}
\newcommand{\labtarft}[1]{\hat{{\sf tar}}_{#1}}

\newcommand{\labct}[2]{\hat{{\sf col}}_{#1}(#2)}
\newcommand{\labinit}[2]{\hat{{\sf ini}}_{#1}(#2)}
\newcommand{\labtart}[2]{\hat{{\sf tar}}_{#1}(#2)}

\newcommand{\ini}{0}
\newcommand{\tar}{r}
\newcommand{\seq}[1]{\mathcal{#1}}

\newcommand{\Rgraph}[2]{{\sf R}_{#1}^{#2}}

\newcommand{\hset}[1]{\{ #1\}}
\newcommand{\hseq}[1]{\langle #1\rangle}

\newcounter{one}
\newcounter{two}
\newcounter{three}
\newcounter{four}
\newcounter{five}
\newcounter{six}
\newenvironment{listing}[1]{%
        \begin{list}{*}{%
                 \settowidth{\labelwidth}{#1}%
                 \setlength{\leftmargin}{\labelwidth}%
                  \advance \leftmargin by 12pt
                   \setlength{\itemsep}{0pt}%
                   \setlength{\parsep}{0pt}%
                   \setlength{\topsep}{0pt}%
                   \setlength{\parskip}{0pt}%
}%
}{%
\end{list}}

%% file: appendix.tex
\newpage
\appendix

\section{Proofs Omitted from Section~\ref{sec:algorithm}}

	We first introduce some notation.
	
	For $i \ge 2$, let $\mapj$ be any node (i.e., a $\numk$-list coloring of $G_i$) in the reconfiguration graph $\Ri{i-1}$.
	Recall that the vertex $v_i$ is adjacent with only the spine vertex $\spine{i-1}$ in $G_i$. 
	Let $c$ be any color in $\verlist(v_{i}) \setminus \hset{\mapf(\spine{i-1})}$.
	Then, we say that $\mapj$ can be \textit{extended by} $c$ to a $\numk$-list coloring $\mapi$ of $G_i$ such that $\mapi(v_i) = c$ and $\mapi(v) = \mapj(v)$ for all vertices $v \in V_{i-1}$;
we simply denote such an extension by $\mapj + c=\mapi$.
	Note that $\mapi$ is a node in $\Ri{i}$. 

	For any e-node $y$ in the encoding graph $\Enc{i-1}$ of $\Rizero{i-1}$, recall that $\mapset{i-1}{y}$ is the set of all nodes in $\Rizero{i-1}$ that were contracted into $y$. 
	For a color $c \in \verlist(v_{i})\setminus \hset{\labc{i-1}{y}}$, let $\extend{i-1}{y}{c} =\hset{\mapj + c:\mapj \in \mapset{i-1}{y}}$, that is, $\extend{i-1}{y}{c}$ is the set of nodes in $\Ri{i}$ that are extended by $c$ from nodes in $\mapset{i-1}{y}$ $\bigl( \subseteq V(\Rizero{i-1}) \bigr)$.


\subsection{Proof of Lemma~\ref{lem:leaf}}
	
	In this subsection, we prove Lemma~\ref{lem:leaf}.
	Therefore, suppose that we have already computed $(\Enc{i-1}, \labcf{i-1}, \labinif{i-1}, \labtarf{i-1})$ for $i \ge 2$, and assume that $v_i\in \Vleaf$ and $\verlist(v_i) = \{ c_1, c_2 \}$.

	We first prove that $V(\Enc{i}) \subseteq V(\Encp{i})$ holds, as in the following lemma.
	\begin{lemma} \label{lem:leaf1}
	Let $x$ be any e-node in $\Enc{i}$.
	Then, there exists an e-node $\xhat$ in $V(\Encp{i}) = V(\Encsub{i-1}{c_1}) \cup V(\Encsub{i-1}{c_2})$ such that $\bigcup \bigl\{ \extend{i-1}{\xhat}{c} : c \in \verlist(v_{i}) \setminus \hset{\labc{i-1}{\xhat}} \bigr\} = \mapset{i}{x}$.
	\end{lemma}
	\begin{proof}
	Let $\mapi$ be any node in $\mapset{i}{x} \subseteq V(\Rizero{i})$. 
	Then, by Lemma~\ref{lem:restrict} the node $\mapi[V_{i-1}]$ is contained in $\Rizero{i-1}$, and hence there exists an e-node $\xhat$ in $\Enc{i-1}$ such that $\mapi[V_{i-1}] \in \mapset{i-1}{\xhat}$. 
	Notice that, since $\verlist(v_i) = \{ c_1, c_2 \}$, we have  $V(\Encsub{i-1}{c_1}) \cup V(\Encsub{i-1}{c_2})=V(\Enc{i-1})$.
	Thus, the e-node $\xhat$ is contained in $\Encp{i}$, too. 
	Therefore, we prove that $\bigcup \bigl\{ \extend{i-1}{\xhat}{c} : c \in \verlist(v_{i}) \setminus \hset{\labc{i-1}{\xhat}} \bigr\} = \mapset{i}{x}$.
	
	Let $\mapip$ be any node in $\mapset{i}{x}$. 
	Then, $\mapi \eqrel{\spine{i}} \mapip$, and hence $\Rizero{i}$ contains a path $\hseq{\mapi_1, \mapi_2, \ldots, \mapi_\ell}$ such that $\mapi_1 = \mapi$, $\mapi_{\ell} = \mapip$ and $\mapi_j(\spine{i}) = \mapi(\spine{i}) = \mapip(\spine{i})$ for all $j \in \{1,2,\ldots,\ell\}$. 
	Thus, we can obtain a path $\hseq{\mapi_1[V_{i-1}], \mapi_2[V_{i-1}], \ldots, \mapi_\ell[V_{i-1}]}$ in which the spine vertex $\spine{i}$ always receives the same color $\mapi(\spine{i}) = \mapip(\spine{i})$;
note that $\mapi_j[V_{i-1}] = \mapi_{j+1}[V_{i-1}]$ may hold if $v_i$ is recolored, but we can simply drop $\mapi_{j+1}[V_{i-1}]$ in such a case.  
	Since $\spine{i} = \spine{i-1}$ for the case where $v_i \in \Vleaf$, we have $\mapi[V_{i-1}] \eqrel{\spine{i-1}} \mapip[V_{i-1}]$. 
	Therefore, $\mapip[V_{i-1}] \in \mapset{i-1}{\xhat}$. 
	Since $\labc{i-1}{\xhat}$ represents the color assigned to $\spine{i} = \spine{i-1}$ and $v_i$ is adjacent with $\spine{i-1}$ in $G_i$, the color $\mapip(v_i)$ is clearly contained in $\verlist(v_{i}) \setminus \hset{\labc{i-1}{\xhat}}$. 
	We thus have $\mapip[V_{i-1}] + \mapip(v_i) = \mapip \in \bigcup \bigl\{ \extend{i-1}{\xhat}{c} : c \in \verlist(v_{i}) \setminus \hset{\labc{i-1}{\xhat}} \bigr\}$. 

	Let $\mapipp$ be any node in $\bigcup \bigl\{ \extend{i-1}{\xhat}{c} : c \in \verlist(v_{i}) \setminus \hset{\labc{i-1}{\xhat}} \bigr\}$ such that $\mapipp = \mapj + c$ for some node $\mapj$ in $\mapset{i-1}{\xhat}$ and $c \in \verlist(v_{i}) \setminus \hset{\labc{i-1}{\xhat}}$.
	Since $\mapj \in \mapset{i-1}{\xhat}$ and $\mapi[V_{i-1}] \in \mapset{i-1}{\xhat}$, we have $\mapi[V_{i-1}] \eqrel{\spine{i-1}} \mapj$ and hence $\Rizero{i-1}$ contains a path $\hseq{\mapj_1, \mapj_2, \ldots, \mapj_\ell}$ such that $\mapj_1 = \mapi[V_{i-1}]$, $\mapi_{\ell} = \mapj$ and $\mapj_j(\spine{i-1}) = \mapj(\spine{i-1})$ for all $j \in \{1,2,\ldots,\ell\}$. 
	Since $\spine{i} = \spine{i-1}$ and $\mapi(v_i) \in \verlist(v_{i}) \setminus \hset{\labc{i-1}{\xhat}}$, the sequence $\hseq{\mapj_1 + \mapi(v_i), \mapj_2 + \mapi(v_i), \ldots, \mapj_\ell+\mapi(v_i)}$ is a path in $\Rizero{i}$. 
	If $\mapi(v_i) \neq c$, we add one more adjacent node $\mapj_{\ell} + c$ to the last. 
	Since $\mapj_1 + \mapi(v_i) = \mapi$ and $\mapj_{\ell}+c = \mapipp$, we thus have $\mapi \eqrel{\spine{i}} \mapipp$ and hence $\mapipp \in \mapset{i}{x}$. 
	\qed
	\end{proof}
	
	By Lemma~\ref{lem:leaf1} we may identify each e-node $x$ in $\Enc{i}$ with the corresponding e-node $\xhat$ in $\Encp{i}$. 

	We then prove the following lemma.
	\begin{lemma} \label{lem:leaf2}
	Let $x$ and $y$ be two e-nodes in $\Enc{i}$, and let $\xhat$ and $\yhat$ be two e-nodes in $\Encp{i}$ corresponding to $x$ and $y$, respectively.
	Then, $\xhat \yhat \in E(\Encp{i}) = E(\Encsub{i-1}{c_1}) \cup E(\Encsub{i-1}{c_2})$ if and only if $xy \in E(\Enc{i})$. 
	\end{lemma}
	\begin{proof}
	We first prove the only-if-part.
	Suppose that $\xhat \yhat \in E(\Encsub{i-1}{c_1})$; it is symmetric for the other case where $\xhat \yhat \in E(\Encsub{i-1}{c_2})$.
	Then, there exist two adjacent nodes $\mapj_x \in \mapset{i-1}{\xhat}$ and $\mapj_y \in \mapset{i-1}{\yhat}$ such that $\mapj_x(\spine{i-1}) \neq c_1$ and $\mapj_y(\spine{i-1}) \neq c_1$.
	Therefore, by Lemma~\ref{lem:leaf1} we have $\mapj_x+c_1\in \mapset{i}{x}$ and $\mapj_y+c_1\in \mapset{i}{y}$.
	Note that, since $\xhat \neq \yhat$, we know that only the spine vertex $\spine{i-1} = \spine{i}$ is recolored between $\mapj_x$ and $\mapj_y$. 
	Thus, $\mapj_x+c_1$ and $\mapj_y+c_1$ are adjacent in $\Rizero{i}$, and hence we have $xy\in E(\Enc{i})$.

	We then prove the if-part.
	Since $xy\in E(\Enc{i})$, there exist two adjacent nodes $\mapi_x \in \mapset{i}{x}$ and $\mapi_y \in \mapset{i}{y}$ in $\Rizero{i}$.
	Since $x \neq y$, only the spine vertex $\spine{i}$ is recolored between $\mapi_x$ and $\mapi_y$, and hence $\mapi_x(v_i) = \mapi_y(v_i)$.
	Therefore, $\mapi_x[V_{i-1}]$ and $\mapi_y[V_{i-1}]$ are adjacent. 
	We assume that $c_1 = \mapi_x(v_i) = \mapi_y(v_i)$ without loss of generality.
	Then, since $v_i$ and $\spine{i}$ are adjacent, $\mapi_x(\spine{i}) \neq c_1$ and $\mapi_y(\spine{i}) \neq c_1$. 
	By Lemma~\ref{lem:leaf1} we have $\mapi_x \in \bigcup \bigl\{ \extend{i-1}{\xhat}{c} : c \in \verlist(v_{i}) \setminus \hset{\labc{i-1}{\xhat}} \bigr\}$;
this implies that $\mapi_x[V_{i-1}] \in \mapset{i-1}{\xhat}$.
	Similarly, $\mapi_y[V_{i-1}] \in \mapset{i-1}{\yhat}$. 
	Since $\mapi_x[V_{i-1}]$ and $\mapi_y[V_{i-1}]$ are adjacent, $\xhat \yhat \in E(\Enc{i-1})$.
	Furthermore, since $\labc{i-1}{\xhat} \neq c_1$ and $\labc{i-1}{\yhat} \neq c_1$, we have $\xhat \yhat \in E(\Encsub{i-1}{c_1})$.
	Therefore, $\xhat \yhat \in E(\Encsub{i-1}{c_1}) \subseteq E(\Encp{i})$. 
%
%
%
	\qed
	\end{proof}

	We now prove the following lemma.
%
	\begin{lemma} \label{lem:leaf3}
	$\Enct{i} = \Enc{i}$.
	\end{lemma}
	\begin{proof}
	Recall that $\Enc{i}$ consists of a single connected component which contains the e-node $z$ such that $\mapf_{\ini}[V_{i}] \in \mapset{i}{z}$.
	Consider the set of all e-nodes $\xhat$ in $\Encp{i}$ that correspond to the e-nodes $x$ in $\Enc{i}$.
	By Lemma~\ref{lem:leaf2} the e-node set forms a connected subgraph of a single connected component in $\Encp{i}$.
	Furthermore, by Lemma~\ref{lem:leaf1} the component in $\Encp{i}$ contains the e-node $\zhat$ such that $\mapf_{\ini}[V_i] \in \mapset{i}{\zhat}$; 
by the construction, $\zhat$ is contained in $\Enct{i}$. 
	We thus have $V(\Enc{i}) \subseteq V(\Enct{i})$.

	Therefore, to show $\Enc{i} = \Enct{i}$, by Lemma~\ref{lem:leaf2} it suffices to prove that there exists no edge $x \yhat \in E(\Enct{i})$ which joins two e-nodes $x \in V(\Enc{i})$ and $\yhat \in V(\Enct{i}) \setminus V(\Enc{i})$.
	Suppose for a contradiction that there exists such an edge $x \yhat \in E(\Enct{i})$.
	By the construction, $x \yhat \in E(\Encsub{i-1}{c_1}) \cup E(\Encsub{i-1}{c_2})$; we may assume that $x \yhat \in E(\Encsub{i-1}{c_1})$ without loss of generality.
	Then, there exist two adjacent nodes $\mapj_x \in \mapset{i-1}{x}$ and $\mapj_{\yhat} \in \mapset{i-1}{\yhat}$ such that  $\mapj_x(\spine{i-1}) \neq c_1$ and $\mapj_{\yhat}(\spine{i-1}) \neq c_1$.
	Therefore $\mapj_x$ and $\mapj_{\yhat}$ can be extended by $c_1$, and $\mapj_x+c_1\in \mapset{i}{x}$ and $\mapj_{\yhat}+c_1\in \mapset{i}{\yhat}$ are adjacent in $\Rizero{i}$.
	By Lemma~\ref{lem:leaf1} we then have $\yhat \in V(\Enc{i})$; this contradicts the assumption that $\yhat \in V(\Enct{i}) \setminus V(\Enc{i})$.
	\qed
	\end{proof}
		
	Finally, we show the following lemma. 
	\begin{lemma} \label{lem:leaf4}
	$\labcft{i} = \labcf{i}$, $\labinift{i}=\labinif{i}$ and $\labtarft{i}=\labtarf{i}$.
	\end{lemma}
	\begin{proof}
	Recall that $\spine{i} = \spine{i-1}$ if $v_i \in \Vleaf$. 
	Then, the lemma follows immediately from Lemma~\ref{lem:leaf1}.
	\qed
	\end{proof}
	
	This completes the proof of Lemma~\ref{lem:leaf}.

\subsection{Proof of Lemma~\ref{lem:spine}}
	
	In this subsection, we prove Lemma~\ref{lem:spine}.
	Therefore, suppose that we have already computed $(\Enc{i-1}, \labcf{i-1}, \labinif{i-1}, \labtarf{i-1})$ for $i \ge 2$, and assume that $v_i\in \Vspine$.
	
	We first prove that $V(\Enc{i}) \subseteq V(\Encp{i})$ holds, as in the following lemma.
	\begin{lemma} \label{lem:spine1}
	Let $x$ be any e-node in $\Enc{i}$ with $\labc{i}{x}=c$.
	Then, there exists exactly one connected component $H$ in $\Encsub{i-1}{c}$ such that $\bigcup \hset{\extend{i-1}{y}{c}:y\in V(H)}=\mapset{i}{x}$.
	\end{lemma}
	\begin{proof}
	Let $\mapi$ be any node in $\mapset{i}{x} \subseteq V(\Rizero{i})$. 
	Then, by Lemma~\ref{lem:restrict} the node $\mapi[V_{i-1}]$ is contained in $\Rizero{i-1}$, and hence there exists an e-node $z$ in $\Enc{i-1}$ such that $\mapi[V_{i-1}] \in \mapset{i-1}{z}$.
	Since $\spine{i} \neq \spine{i-1}$ and $\spine{i}$ is adjacent with $\spine{i-1}$, the assumption $\labc{i}{x}=c$ implies that $\mapi(\spine{i}) = c$ and hence $\mapi(\spine{i-1}) \neq c$.
	Therefore, we have $\labc{i-1}{z} \neq c$, and hence $\Encsub{i-1}{c}$ has exactly one connected component $H$ that contains $z$.
	We thus prove that $\bigcup \hset{\extend{i-1}{y}{c}:y\in V(H)} = \mapset{i}{x}$ for the connected component $H$.
	
	Let $\mapip$ be any node in $\mapset{i}{x}$. 
	Then, $\mapip(\spine{i}) = c$ and hence it suffices to show that $H$ contains an e-node $y$ such that $\mapip[V_{i-1}] \in \mapset{i-1}{y}$. 
	Since $\mapi \eqrel{\spine{i}} \mapip$, there exists a path $\hseq{\mapi_1, \mapi_2, \ldots, \mapi_\ell}$ in $\Rizero{i}$ such that $\mapi_1 = \mapi$, $\mapi_{\ell} = \mapip$ and $\mapi_j(\spine{i}) = \mapi(\spine{i}) = \mapip(\spine{i}) = c$ for all $j \in \{1,2,\ldots,\ell\}$. 
	Since $\spine{i}$ is adjacent with $\spine{i-1}$, $\mapi_j(\spine{i-1}) \neq c$ holds for all $j \in \{1,2,\ldots,\ell\}$.
	Therefore, there exists a connected component $H^\prime$ in $\Encsub{i-1}{c}$ such that the path $\hseq{\mapi_1[V_{i-1}], \mapi_2[V_{i-1}], \ldots, \mapi_\ell[V_{i-1}]}$ is contained in $\bigcup \hset{\mapset{i-1}{y^\prime} : y^\prime \in V(H^\prime)}$.
	Because $z \in V(H)$ and $\mapi_1[V_{i-1}] = \mapi[V_{i-1}] \in \mapset{i-1}{z}$, we have $H^\prime = H$.
	Thus, $H$ contains an e-node $y$ such that $\mapip[V_{i-1}] \in \mapset{i-1}{y}$. 
	Then, we have $\mapip \in \bigcup \hset{\extend{i-1}{y}{c}:y\in V(H)}$.


	Conversely, let $\mapipp$ be any node in $\bigcup \hset{\extend{i-1}{y}{c}:y\in V(H)}$.
	Since $H$ is connected, the subgraph of $\Rizero{i-1}$ induced by  $\bigcup \hset{\mapset{i-1}{y} : y\in V(H)}$ is connected, too.
	Then, the induced subgraph contains a path $\hseq{\mapj_1, \mapj_2, \ldots, \mapj_\ell}$ such that $\mapj_1 = \mapi[V_{i-1}]$ and $\mapj_{\ell} = \mapipp[V_{i-1}]$. 
	Furthermore, since $H$ is a connected component in $\Encsub{i-1}{c}$, we know that $\mapj_j(\spine{i-1}) \neq c$ for all nodes $\mapj_j$, $j \in \{1,2,\ldots,\ell\}$. 
	Therefore, we can extend each node $\mapj_j$ by $c$, and obtain a path $\hseq{\mapj_1+c, \mapj_2+c, \ldots, \mapj_\ell+c}$.
	Since $\mapj_1+c = \mapi$ and $\mapj_\ell+c = \mapipp$, we thus have $\mapi \eqrel{\spine{i}} \mapipp$.
	Since $\mapi \in \mapset{i}{x}$, we have $\mapipp \in \mapset{i}{x}$.
	\qed
	\end{proof}

	For each e-node $x\in V(\Enc{i})$, let $H_x$ be the connected component in $\Encsub{i-1}{\labc{i}{x}}$ which satisfies Lemma~\ref{lem:spine1}.
	Then, we can identify the e-node $x$ in $\Enc{i}$ with the e-node $\xhat$ in $\Encp{i}$ such that $\nset{\xhat}=V(H_x)$.
	
	We then prove the following lemma.
	\begin{lemma} \label{lem:spine2}
	Let $x$ and $y$ be two e-nodes in $\Enc{i}$, and let $\xhat$ and $\yhat$ be two e-nodes in $\Encp{i}$ corresponding to $x$ and $y$, respectively.
	Then, $\nset{\xhat} \cap \nset{\yhat} \neq \emptyset$ if and only if $xy\in E(\Enc{i})$.
	\end{lemma}
	\begin{proof}
	We first prove the only-if-part.
	Suppose that the set $\nset{\xhat} \cap \nset{\yhat}$ contains an e-node $a$ in $\Enc{i-1}$.
	Choose an arbitrary node $\mapj \in \mapset{i-1}{a}$, then two nodes $\mapj+\labc{i}{x} \in \mapset{i}{x}$ and $\mapj+\labc{i}{y} \in \mapset{i}{y}$ are adjacent in $\Rizero{i}$.
	Thus, $xy\in E(\Enc{i})$.

	We then prove the if-part.
	Suppose that $xy\in E(\Enc{i})$, then there exist two adjacent nodes $\mapi_x\in \mapset{i}{x}$ and $\mapi_y\in \mapset{i}{y}$ in $\Rizero{i}$.
	Since $\mapi_x$ and $\mapi_y$ are adjacent and $\labc{i}{x}\neq \labc{i}{y}$, we know that $\mapi_x[V_{i-1}] = \mapi_y[V_{i-1}]$. 
	By Lemma~\ref{lem:spine1}, there exists an e-node $a\in \nset{\xhat}$ such that $\mapi_x[V_{i-1}] \in \mapset{i-1}{a}$.
	Similarly, there exists an e-node $b\in \nset{\yhat}$ such that $\mapi_y[V_{i-1}] \in \mapset{i-1}{b}$.
	Since $\mapi_x[V_{i-1}] = \mapi_y[V_{i-1}]$, we have $a=b$.
	Therefore, $a = b\in \nset{\xhat} \cap \nset{\yhat} \neq \emptyset$.
	\qed
	\end{proof}
	
	By Lemmas \ref{lem:spine1} and \ref{lem:spine2}, we have $\Enc{i}\subseteq \Encp{i}$ and $\labcf{i}=\labcft{i}$.
	
	We now prove the following lemma.
	\begin{lemma} \label{lem:spine4}
	$\labinif{i}=\labinift{i}$ and $\labtarf{i}=\labtarft{i}$.
	\end{lemma}
	\begin{proof}
	We prove only $\labinif{i}=\labinift{i}$; it is similar to prove $\labtarf{i}=\labtarft{i}$.

	Let $x$ be any e-node in $\Enc{i}$ such that $\labini{i}{x}=1$.
	Then, $\mapf_\ini[V_i]\in \mapset{i}{x}$, and $\labc{i}{x}=\labct{i}{x}=\mapf_\ini(v_i)$ holds.
	By Lemma~\ref{lem:spine1} there exists an e-node $y\in \nset{x}$ such that $\mapf_\ini[V_{i-1}]\in \mapset{i-1}{y}$ and $\labini{i-1}{y}=1$.
	Since $\mapf_\ini[V_{i}] = \mapf_{\ini}[V_{i-1}] + \labc{i}{x}$, we thus have $\labinit{i}{x}=1$.
	
	Conversely, let $\xhat$ be any e-node in $\Encp{i}$ such that $\labinit{i}{\xhat}=1$.
	Then, $\nset{\xhat}$ contains an e-node $y$ such that $\mapf_\ini[V_{i-1}]\in \mapset{i-1}{y}$ and $\labini{i-1}{y}=1$.
	Note that $y$ is in $\Encsub{i-1}{\labc{i}{\xhat}}$, and $\mapf_\ini[V_{i-1}]+\labc{i}{\xhat}=\mapf_\ini[V_i]$.
	By Lemma~\ref{lem:spine1} we thus have $\mapf_\ini[V_{i}] \in \mapset{i}{\xhat}$ and hence $\labini{i}{\xhat}=1$.
	\qed
	\end{proof}
	
	Finally, we show following lemma.
	\begin{lemma} \label{lem:spine3}
	$\Enct{i} = \Enc{i}$.
	\end{lemma}
	\begin{proof}
	Recall that $\Enc{i}$ consists of a single connected component which contains the e-node $z$ such that $\mapf_{\ini}[V_{i}] \in \mapset{i}{z}$.
	By Lemma~\ref{lem:spine2} $\Enc{i}$ is contained in one connected component of $\Encp{i}$ as a subgraph, and the connected component contains an e-node $\zhat$ such that $\mapf_{\ini}[V_{i}] \in \mapset{i}{\zhat}$.
	By Lemma \ref{lem:spine4} we have $\labinit{i}{\zhat}=1$, and hence the connected component is indeed $\Enct{i}$.
	We thus have $\Enc{i} \subseteq \Enct{i}$.
	
	Therefore, to show $\Enc{i} = \Enct{i}$, by Lemma~\ref{lem:spine2} it suffices to prove that there exists no edge $x \yhat \in E(\Enct{i})$ which joins two e-nodes $x \in V(\Enc{i})$ and $\yhat \in V(\Enct{i}) \setminus V(\Enc{i})$.
	Suppose for a contradiction that there exists such an edge $x \yhat \in E(\Enct{i})$.
	Then, the set $\nset{x} \cap \nset{\yhat}$ contains an e-node $z$.
	Choose an arbitrary node $\mapj \in \mapset{i-1}{z}$, then two nodes $\mapj+\labc{i}{x}$ and $\mapj+\labc{i}{\yhat}$ are adjacent in $\Rizero{i}$.
	Then, $\mapj+\labc{i}{\yhat} \in \bigcup_{z^\prime\in \nset{\yhat}}(\extend{i-1}{z^\prime}{\labc{i}{\yhat}})$.
	By Lemma~\ref{lem:spine1} the corresponding e-node $y$ should be contained in $\Enc{i}$; this contradicts the assumption that $\yhat \in V(\Enct{i}) \setminus V(\Enc{i})$.
	\qed
	\end{proof}	
	
	This completes the proof of Lemma \ref{lem:spine}.

	\subsection{Proof of Lemma~\ref{lem:size}}
	
	In this subsection, we prove Lemma~\ref{lem:size}.
	Since $|V(\Enct{i})| \le |V(\Encp{i})|$ holds, it suffices to prove the following inequality:
for each index $i \in \{1,2,\ldots, n\}$, 
	\begin{equation} \label{eq:size}
		|V(\Encp{i})| \le \left\{
				\begin{array}{ll}
				2 & ~~~\mbox{if $i = 1$}; \\
				|V(\Enc{i-1})| + \degG(v_i) & ~~~\mbox{otherwise}.
				\end{array} \right.
	\end{equation}
	
	Consider the case where $v_i$ is a leaf. 
	Then, $V(\Encp{i}) = V(\Enc{i-1})$, and hence Eq.~(\ref{eq:size}) clearly holds. 
	In the remainder of this subsection, we thus consider the case where $v_i$ is a spine vertex.
	
	For a graph $G=(V,E)$, we denote by $\comp{G}$ the number of connected components in $G$.
	For a connected graph $G$, that is, $\comp{G}=1$, we denote by $\stree{G}$ any spanning tree of $G$.
	Since $E(\stree{G}) \subseteq E(G)$, we clearly have the following proposition.
%
%
	\begin{proposition} \label{pro:span}
	Let $G$ be a connected graph, and let $V_0$ be any vertex subset of $G$.
	Then, $\comp{G[V_0]}\le \comp{\stree{G}[V_0]}$ holds.
	\end{proposition}
	
	We now apply Case (B) of our algorithm to any spanning tree $\stree{\Enc{i-1}}$ of $\Enc{i-1}$, instead of applying the operation to $\Enc{i-1}$.
	Let $\Encst{i}$ be the obtained encoding graph, instead of $\Encp{i}$.
	Then, we have the following lemma.
	\begin{lemma} \label{lem:size1}
	$|V(\Encp{i})| \le |V(\Encst{i})|$.
	\end{lemma}
	\begin{proof}
	For each color $c \in \verlist(v_i)$, let $\Encsubst{i-1}{c}$ be the subgraph of $\stree{\Enc{i-1}}$ obtained by deleting all e-nodes $y$ in $\stree{\Enc{i-1}}$ with $\labc{i-1}{y} = c$.
	Then,
	\[
		|V(\Encp{i})|=\sum_{c\in \verlist(v_{i})}\comp{\Encsub{i-1}{c}},
	\]
and
	\begin{equation} \label{eq:size2}
		|V(\Encst{i})|=\sum_{c\in \verlist(v_{i})}\comp{\Encsubst{i-1}{c}}.
	\end{equation}
	By Proposition~\ref{pro:span} we have $\comp{\Encsub{i-1}{c}}\le \comp{\Encsubst{i-1}{c}}$ for each color $c\in \verlist(v_{i})$, and hence $|V(\Encp{i})|\le |V(\Encst{i})|$.
	\qed
	\end{proof}
	
	We finally show the following lemma, which verifies Eq.~(\ref{eq:size}) and hence completes the proof of Lemma~\ref{lem:size}.
	\begin{lemma}
	If $v_i$ is a spine vertex, then $|V(\Encp{i})|\le |V(\Enc{i-1})| + \degG(v_i)$.
	\end{lemma}
	\begin{proof}
	We first consider the case where $|V(\Enc{i-1})| = 1$.
	Then, $\comp{\Encsub{i-1}{c}} \le 1$ for any color $c\in \verlist(v_i)$, and hence 
	\[
		|V(\Encp{i})| =\sum_{c\in \verlist(v_{i})}\comp{\Encsub{i-1}{c}} \le |\verlist(v_i)|.
	\]
	By Lemma~\ref{lem:transform} we have $|\verlist(v_i)| \le \degG(v_i) + 1$, and hence 
	\[
		|V(\Encp{i})| \le 1+\degG(v_i) = |V(\Enc{i-1})| + \degG(v_i).
	\]
	
	We then consider the case where $|V(\Enc{i-1})|\ge 2$.
	Recall that $\stree{\Enc{i-1}}$ is a spanning tree of $\Enc{i-1}$, and hence $V(\stree{\Enc{i-1}}) = V(\Enc{i-1})$. 
	For each color $c \in \verlist(v_i)$, let $\Xic{i-1}{c}=\hset{x\in V(\stree{\Enc{i-1}}):\labc{i-1}{x}=c}$.
	For each vertex $x \in V(\stree{\Enc{i-1}})$, we denote by $\degG(\stree{\Enc{i-1}},x)$ the degree of $x$ in $\stree{\Enc{i-1}}$. 
	Then, by deleting $x$ from $\stree{\Enc{i-1}}$, the number of connected components in the resulting graph is increased by $\degG(\stree{\Enc{i-1}}, x) - 1$. 
	We thus have
	\begin{eqnarray} 
		\comp{\Encsubst{i-1}{c}} &=& \comp{\stree{\Enc{i-1}}} +\sum \Bigl\{ \degG(\stree{\Enc{i-1}}, x)-1 : x\in \Xic{i-1}{c} \Bigr\} \nonumber \\
			&=& 1+\sum \Bigl\{ \degG(\stree{\Enc{i-1}}, x)-1 : x\in \Xic{i-1}{c} \Bigr\}. \label{eq:size3}
	\end{eqnarray}
	By Lemma~\ref{lem:size1} and Eq.~(\ref{eq:size2}) we have
	\begin{eqnarray*}
		|V(\Encp{i})| &\le& |V(\Encst{i})| = \sum_{c\in \verlist(v_{i})}\comp{\Encsubst{i-1}{c}}.
	\end{eqnarray*}
	Therefore, by Eq.~(\ref{eq:size3}) we have
	\begin{eqnarray}
			|V(\Encp{i})| &\le& \sum_{c\in \verlist(v_{i})} \left(1+\sum \Bigl\{ \degG(\stree{\Enc{i-1}}, x)-1 : x\in \Xic{i-1}{c} \Bigr\} \right) \nonumber \\
							&\le& |\verlist(v_i)|+\sum \Bigl\{ \degG(\stree{\Enc{i-1}}, x)-1 : x\in V(\stree{\Enc{i-1}})  \Bigr\} \nonumber \\
							&=& |\verlist(v_i)|+2|E(\stree{\Enc{i-1}})|-|V(\stree{\Enc{i-1}})|. \label{eq:size4}
	\end{eqnarray}
	Since $\stree{\Enc{i-1}}$ is a tree, $|E(\stree{\Enc{i-1}})| = |V(\stree{\Enc{i-1}})|-1$.
	Furthermore, recall that $\stree{\Enc{i-1}}$ is a spanning tree of $\Enc{i-1}$, and hence $V(\stree{\Enc{i-1}}) = V(\Enc{i-1})$. 
	By Eq.~(\ref{eq:size4}) we thus have $|V(\Encp{i})| \le |\verlist(v_i)|+|V(\Enc{i-1})|-2$. 
	By Lemma~\ref{lem:transform} we have $|\verlist(v_i)| \le \degG(v_i) +1$, and hence 
	\[
		|V(\Encp{i})| \le |\verlist(v_i)|+|V(\Enc{i-1})|-2\le |V(\Enc{i-1})|+\degG(v_i)-1,
	\]
as required.
	\qed
	\end{proof}